\documentclass[reprint,amsmath,amssymb,aps,pra]{revtex4-2}

\usepackage{graphicx}
\usepackage{dcolumn}
\usepackage{bm}
\usepackage{hyperref}
\usepackage{braket}
\usepackage{amsmath}
\usepackage{amssymb}
\usepackage{amsthm}
\usepackage{quantikz}
\usepackage[ruled,vlined,linesnumbered]{algorithm2e}

\newtheorem{claim}{Claim}

\usepackage{xcolor}

\begin{document}

\title{A Logarithmic-Depth Quantum Multiplier}
\author{Fred Sun${}^{1,2}$}
\author{Anton Borissov${}^{1}$}
\affiliation{${}^{1}$softwareQ inc., Waterloo, Ontario, N2L 0C7, Canada}
\affiliation{${}^{2}$University of Waterloo, Waterloo, Ontario, N2L 3G1, Canada}
\date{\today}
\begin{abstract}
We present a quantum algorithm for multiplying two $n$-bit integers with overall circuit depth and $T$-depth both bounded by $O(\log n)$, while using $O(n^2)$ gates and ancillary qubits. Our construction generates partial products via indicator-controlled copying and adds them using a tree of reversible $3$-to-$2$ adders followed by one quantum carry-lookahead adder. Building on the carry-save adders of Kim et al., we combine full input fanout, workspace reuse, and phase-corrected gate implementations to reduce $T$-depth. We provide explicit circuit constructions, proofs of correctness, and concrete closed-form resource bounds.
\end{abstract}

\maketitle
\section{Introduction}

Quantum arithmetic is a fundamental building block in many quantum algorithms, most notably the modular exponentiation component of Shor's factoring algorithm \cite{Shor1997,pavlidis2013fastquantummodularexponentiation}, but also in scientific computing applications, including quantum algorithms for linear systems \cite{Childs_2017,Harrow_2009} and Hamiltonian simulation using truncated Taylor series \cite{Berry_2015}. Efficient implementations of elementary arithmetic operations, e.g.\ addition and multiplication, play a critical role in the optimization of large-scale quantum circuits
\cite{wang2024comprehensivestudyquantumarithmetic, cuccaro2004newquantumripplecarryaddition,Draper2004,Ruiz_Perez_2017, gidney2025classicalquantumadderconstantworkspace}.

In fault-tolerant quantum computing, circuit depth and, in particular,
$T$-depth are important measures of the cost of an algorithm, in addition
to total gate count \cite{Selinger2013,Amy2012}. Since non-Clifford gates
can dominate the overhead of fault-tolerant implementations, optimizing
$T$-depth is often desirable, even at the expense of qubit count or
Clifford gate depth.

In this work, we study the multiplication of two $n$-bit natural numbers.
We present a simple and modular quantum multiplication algorithm that achieves $O(\log n)$ depth and $T$-depth using the standard Clifford+$T$ gate set. The algorithm prioritizes optimizing $T$-depth, trading off a quadratic number of clean ancillary qubits and gates to allow more operations to run in parallel.

The construction has three main components: copying the inputs, generating the partial products, and adding the partial products. We use enough copies of the input bits to compute every partial product in parallel. The copies are then cleared so that their qubits can be reused during addition. A tree of $3$-to-$2$ adders reduces the partial products to two numbers, and one final binary addition gives the product. Reversing the intermediate computations clears the workspace. 

\subsection{Background}

\begin{table*}[t]
\centering
\normalfont\normalsize
\caption{Asymptotic resource costs for quantum multiplication algorithms.}
\label{table:compare_asymp}
\begin{tabular}{|c|c|c|c|}
\hline
Paper & Depth & Gate count & Ancillary qubits\\\hline
Shift-and-add \cite{Munoz-Coreas2019}
& $O(n^2)$ & $O(n^2)$ & $O(n)$\\\hline
Karatsuba \cite{parent2017improvedreversiblequantumcircuits}
& $O(n^{1.158})$ & $O(n^{\log_2 3})$ & $O(n^{1.428})$\\\hline
Toom--Cook 8-way \cite{Putranto2023}
& $O(n^{1.057})$ & $O(n^{\log_8 15})$ & $O(n^{1.246})$\\\hline
Sch\"onhage--Strassen$^{*}$ \cite{Nie2023}
& $O(\log^2 n)$ & $O(n\log n\log\log n)$ & $O(n\log n\log\log n)$\\\hline
Carry-save multiplication \cite{Gossett1998, Kim2025}
& $O(\log n)$ & $O(n^2)$ & $O(n^2)$\\\hline
\textbf{This work}
& $O(\log n)$ & $O(n^2)$ & $O(n^2)$\\\hline
\end{tabular}\\
    \footnotesize
    $^{*}$The quantum Schönhage--Strassen algorithm is a galactic algorithm \cite{Nie2023}.
\end{table*}

Many quantum circuits for integer multiplication have been proposed, each with
trade-offs among depth, $T$-count, and ancillary qubit count. In this work, we
compare against representative constructions, as surveyed by Wang et al.
\cite{wang2024comprehensivestudyquantumarithmetic}, that serve as benchmarks in
fault-tolerant settings: shift-and-add multiplication, Karatsuba's algorithm,
Toom--Cook multiplication, and divide-and-conquer approaches such as the quantum
Sch\"onhage--Strassen algorithm
\cite{Munoz-Coreas2019,parent2017improvedreversiblequantumcircuits,
Putranto2023,Nie2023}.

The shift-and-add multiplication method, implemented by Mu\~noz-Coreas and Thapliyal, uses Toffoli gates and controlled adders to construct the standard schoolbook multiplication procedure \cite{Munoz-Coreas2019}. This approach has quadratic $T$-depth and $T$-count, but requires only $O(n)$ ancillary qubits. The quantum implementations of Karatsuba's algorithm and the Toom--Cook multiplication algorithm reduce asymptotic gate complexity by recursively decomposing large multiplications into smaller subproblems \cite{parent2017improvedreversiblequantumcircuits,Dutta2018,Putranto2023}. These approaches improve gate count and depth compared to shift-and-add multiplication at the tradeoff of a greater number of ancillary qubits. Low-space constructions provide another tradeoff; Kahanamoku-Meyer and Yao study fast multiplication with zero ancillas \cite{KahanamokuMeyer2024}.

Nie et al.~\cite{Nie2023} introduce a quantum Sch\"onhage--Strassen algorithm with $O(\log^2 n)$ depth and $O(n\log n\log\log n)$ gate and qubit counts, using a recursive structure based on the Fast Fourier Transform. The same work gives a 
simplified construction with $O(n^{1.5})$ gate and qubit counts, and discusses a binary-adder-tree implementation of textbook multiplication. In a binary adder tree, each level performs complete binary additions. Using logarithmic-depth adders at logarithmically many levels gives an $O(\log^2 n)$ depth bound.

Carry-save arithmetic avoids the need for a full-width binary addition at each level of the addition tree. Gossett introduced quantum carry-save arithmetic and used a tree of $3$-to-$2$ adders to construct a logarithmic-depth multiplier with quadratic qubit count \cite{Gossett1998}. 
Each $3$-to-$2$ adder reduces three equally weighted bits to two bits that contribute to the remaining sum. Each layer performs these reductions in parallel, passing fewer bits to the next layer while requiring only constant depth. 
A suitable tree arrangement leaves two numbers after $O(\log n)$ layers. One final full-width binary addition then gives the total sum.

Following Gossett's work, Kim et al.~\cite{Kim2025} improved the
full-adder circuit to require one Toffoli gate and five CNOTs. They arranged these adders into quantum Wallace and Dadda trees,
followed by a final binary addition, and used uncomputation to clear and reuse ancillary qubits.

We build on these constructions by using copies of the inputs to compute all pairwise bit products in parallel, then reuse the copying workspace during addition. 
Using relative-phase gates \cite{Selinger2013}, we implement each Toffoli layer in one $T$ layer. The additional phases are cancelled during uncomputation or corrected within the final adder, giving a fully unitary circuit.

These choices reduce the $T$-depth of carry-save multiplication while retaining quadratic gate and qubit counts. Table~\ref{table:compare_asymp} summarizes the asymptotic bounds, and Table~\ref{table:compare} compares our resource costs with Kim et al.'s Dadda construction using TD2 and the Draper adder \cite[Table~6]{Kim2025}. 
We achieve lower $T$-depths at every input size tested, at the cost of additional qubits and $T$ gates. For $n=32,64$, the depth reduction also gives a lower $T$-depth--space volume, without requiring intermediate measurements.

\begin{table*}[t]
    \centering
    \renewcommand{\arraystretch}{1.6}
    \caption{Comparison with Kim et al.'s Dadda--Draper TD2
    circuit \cite{Kim2025}. Their circuit uses measurements;
    ours is unitary. Our $D_T$ values are upper bounds,
    and $Q$ counts all qubits.}
    \label{table:compare}
    \begin{tabular}{|c|ccc|ccc|}
    \hline
    & \multicolumn{3}{c|}{Kim et al. \cite{Kim2025}}
    & \multicolumn{3}{c|}{This work}\\\hline
    $n$ & $T$-Depth & Qubits & $D_T \times Q$
        & $T$-Depth & Qubits & $D_T \times Q$\\\hline
    8  & 23  & \textbf{166}  & \textbf{3{,}818}
       & \textbf{19} & 272   & 5{,}168\\
    16 & 41  & \textbf{602}  & \textbf{24{,}682}
       & \textbf{25} & 1056  & 26{,}400\\
    32 & 75  & \textbf{2246} & 168{,}450
       & \textbf{31} & 4160  & \textbf{128{,}960}\\
    64 & 141 & \textbf{8610} & 1{,}214{,}010
       & \textbf{37} & 16512 & \textbf{610{,}944}\\\hline
    \end{tabular}
\end{table*}

The remainder of the paper is organized as follows. In Section~2, we provide
the foundations for our algorithm, including the design and analysis of all
submodules used. In Section~3, we provide the full fast quantum multiplication
algorithm and analyze its resource requirements.

\section{Preliminaries}\label{sec:prelim}
\subsection{Notation}\label{subsec:math}

Consider two $n$-bit nonnegative integers, $x,y$. We express their binary
representation as
\begin{align}
    x &= x_{n-1}\cdot2^{n-1}+\dots+x_1\cdot2+x_0, \label{eq:DefnX}\\
    y &= y_{n-1}\cdot2^{n-1}+\dots+y_1\cdot2+y_0,
\end{align}
where $n$ is the register width, $0\le x,y<2^n$, and for $0\le i\le n-1$,
$x_i,y_i\in\{0,1\}$. The most significant bit of $x$ is $x_{n-1}$ and the
least significant bit is $x_0$. All logarithms are base two unless otherwise stated.

We structure our algorithm around the following representation of the product:
\begin{align}
    xy &= y_{n-1}(2^{n-1}x)+\dots+y_1(2x)+y_0x. \label{Eq1}
\end{align}
This allows us to think of each $y_i$ as a binary indicator variable for whether
$2^ix$ is included in the sum to compute the product.

We find it useful to introduce $n$-bit \emph{partial products},
$\alpha^{(0,i)}:=y_i x$. Their individual bits are $p_{ji}=x_jy_i$, so
\begin{align}
    xy&=\sum_{i=0}^{n-1}2^i\alpha^{(0,i)}
      =\sum_{i=0}^{n-1}\sum_{j=0}^{n-1}2^{i+j}p_{ji}.
    \label{eq:DefnAlpha}
\end{align}
We say that bit $p_{ji}$ has \textit{weight} $2^{i+j}$. We place bits of the same weight
in the same column and add them until two $2n$-bit numbers, $S$ and $C$,
remain. Their sum is the desired product, $S+C=xy$.

Throughout the paper, if we have an $m$-bit nonnegative integer,
$a=a_{m-1}2^{m-1}+\cdots+a_0$, we encode it into an $m$-qubit register as
\begin{align}
    \ket{a}\equiv\ket{a_{m-1}}\cdots\ket{a_0}. \label{eq:MbitQubit}
\end{align}

We would like to find an efficient circuit that implements the following
transformation,
\begin{align}
    \ket{x}\ket{y}\ket{0}^{\otimes2n}\ket{0}_{anc}
    \xrightarrow[\qquad]{}\ket{x}\ket{y}\ket{xy}\ket{0}_{anc},
    \label{eq:Desired}
\end{align}
where $\ket{x}$ and $\ket{y}$ represent $n$-qubit registers, and $\ket{xy}$
represents a $2n$-qubit register, defined using Eq.~\eqref{eq:MbitQubit}.
$\ket{0}_{anc}$ represents a helper ancilla register of size $A$.

To implement the above algorithm, it is useful to define $n$ $n$-qubit copied registers of $\ket{x}$ and $\ket{y}$, namely $\ket{x^{(i)}}$ and $\ket{y^{(i)}}$, for $i=0,\ldots,n-1$. 
The original inputs are $\ket{x^{(0)}}=\ket{x}$ and $\ket{y^{(0)}}=\ket{y}$; the remaining registers are initialized to zero and will hold computational-basis copies of the inputs. We use the $n$-qubit registers $\ket{\alpha^{(0,i)}}$ to store the partial products $\alpha^{(0,i)}=y_i x$.

For the remainder of this section, we describe the three submodules needed to compute the product: copying the inputs, generating the partial products, and adding the partial products. The addition tree reduces the weighted partial products to two $2n$-bit numbers, $S$ and $C$, satisfying $S+C=xy$. One final binary addition then gives the product.

We write $N_{\mathrm{CNOT}}$ and $N_{\mathrm{Tof}}$ for the CNOT and Toffoli counts, and $D$ and $D_{\mathrm{Tof}}$ for the circuit depth and Toffoli depth before Clifford+$T$ decomposition.
After decomposition, $N_T$ denotes the total number of $T$ and
$T^\dagger$ gates, and $D_T$ denotes the $T$-depth. $A$ denotes the number of helper ancillas and $Q$ the total number of qubits. Including the two $n$-qubit inputs and the $2n$-qubit product register gives $Q=4n+A$.
We derive the resource costs when analyzing the submodules below.

\subsection{Fast Copying}\label{subsec:fast_copy}

The \texttt{fast\_copy} algorithm, shown in Algorithm~\ref{algo:fastcopy}, produces $n$ computational-basis copies of an $n$-qubit register in $O(\log n)$ depth, using $O(n^2)$ ancillary qubits. In our multiplication algorithm, \texttt{fast\_copy} creates $n$ copies of both $\ket{x}$ and $\ket{y}$ to increase the number of operations that can be performed in parallel.

\begin{algorithm}[t]
\caption{fast\_copy}\label{algo:fastcopy}
\KwIn{$\ket{x^{(0)}}\ket{0}^{\otimes n(n-1)}$}
\For{$t\gets0$ \KwTo $\lceil\log n\rceil-1$}{
  \For{$i\gets0$ \KwTo $\min(2^t,n-2^t)-1$ \textnormal{in parallel}}{
    Copy $\ket{x^{(i)}}$ onto $\ket{x^{(i+2^t)}}$ using $n$ parallel CNOTs\;
  }
}
\KwOut{$\ket{x^{(0)}}\ket{x^{(1)}}\cdots\ket{x^{(n-1)}}$}
\end{algorithm}

A CNOT gate from a basis state $\ket{a}$, $a\in\{0,1\}$, targeting $\ket{0}$ copies the bit $a$. We can copy an entire register in depth one using $n$ CNOTs, one from each source qubit to the corresponding target qubit.

\begin{claim}
Algorithm~\ref{algo:fastcopy} produces $n$ computational-basis copies of an
$n$-qubit register using $n(n-1)$ CNOTs and depth $\lceil\log n\rceil$.
\end{claim}
\begin{proof}
We start with one populated register and $n-1$ zero registers. At the first timestep, we copy the populated register onto a zero register. At the next timestep, both populated registers can be copied in parallel. In general, every previously populated register can be used as a source at the next timestep, so the number of populated registers doubles.

After $\lceil\log n\rceil$ timesteps, all $n$ registers have been filled. If $n$ is not a power of two, the last timestep copies only as many registers as are still needed. The source and target registers in each timestep are disjoint, so each timestep has depth one. Each of the $n-1$ additional registers is written once, using $n$ CNOTs.
\end{proof}

For an 8-register case, the \texttt{fast\_copy} algorithm is shown in
Fig.~\ref{fig:fast_copy}.

\begin{figure}[h]
    \centering
    \includegraphics[width=\columnwidth]{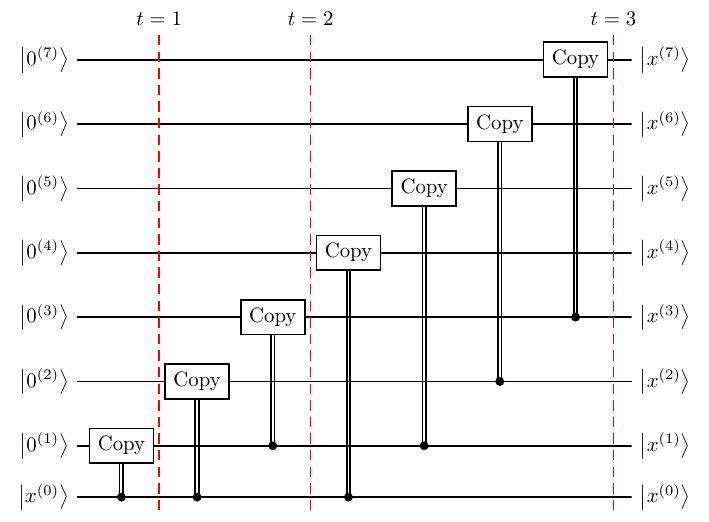}
    \caption{Implementation for fast-copying of 8 registers. The controlled copy gates indicate using $n$ CNOTs in parallel to copy the basis state of the entire register. All operations in each timestep $t$ are performed in parallel.}
    \label{fig:fast_copy}
\end{figure}

On a superposition, the circuit acts as
\begin{align}
    \sum_x c_x\ket{x}\ket{0}^{\otimes n(n-1)}
    \longmapsto\sum_x c_x\ket{x}^{\otimes n}.
\end{align}
The copies are therefore generally entangled. They are computational-basis
copies, not independent copies of an arbitrary quantum state.
Reversing the CNOT network clears the additional registers as long as
their computational-basis values have not changed.

\subsection{Partial Products}\label{subsec:partial_prod}

In this subsection, we introduce the \texttt{conditional\_copy} and
\texttt{partial\_products} subroutines, which compute all partial products,
$\ket{\alpha^{(0,i)}}=\ket{y_ix}$.

First consider one partial product. A Toffoli gate implements
\begin{align}
    \ket{x_j}\ket{y_i}\ket{0}
    \longmapsto\ket{x_j}\ket{y_i}\ket{x_jy_i}.
\end{align}
Thus $y_i$ controls whether bit $x_j$ is copied into the target. Applying this operation to all $n$ bits of $x$ gives $y_ix$. If each Toffoli is assigned its own copies of the controls, all $n$ gates act on disjoint wires and can be performed in one layer. This is the \texttt{conditional\_copy} subroutine.

The \texttt{partial\_products} subroutine performs \texttt{conditional\_copy} for every $i$ in parallel. There are $n$ copies of each input bit, exactly enough to assign one copy to each of the $n$ products in which that bit occurs. One explicit assignment is given below. 

\begin{algorithm}[t]
\caption{partial\_products}\label{algo:itwopowermultiples}
\KwIn{$n$ copies of each input; $n^2$ target qubits initialized to zero}
\For{$i\gets0$ \KwTo $n-1$ \textnormal{in parallel}}{
    \For{$j\gets0$ \KwTo $n-1$ \textnormal{in parallel}}{
        Apply $\text{Toffoli}(x_j^{(i)},y_i^{(j)};p_{ji})$\;
    }
}
\KwOut{Unchanged input copies and $\ket{\alpha^{(0,i)}}=\ket{y_ix}$
for $i=0,\ldots,n-1$}
\end{algorithm}

\begin{claim}
Algorithm~\ref{algo:itwopowermultiples} computes all $n$ partial products
using $n^2$ Toffolis in one layer.
\end{claim}
\begin{proof}
Each Toffoli performs $p_{ji}=x_jy_i$ onto the corresponding target bit. Each control wire $x_j^{(i)}$ and $y_i^{(j)}$ is used exactly once, and all targets are distinct. Thus all $n^2$ Toffolis act on disjoint wires and are able to be run in parallel.
\end{proof}

This assignment only chooses which physical copy controls each gate; it requires no arithmetic or swaps in the quantum circuit. Since none of the controls is changed, the copying networks can be reversed immediately after the partial products have been computed. Their qubits are then available for the addition stage.

\subsection{Parallel Addition of $n$ Terms}\label{subsec:add}

In this subsection, we introduce the \texttt{parallel\_adder\_tree},
which computes
\begin{align}
    xy=\sum_{i=0}^{n-1}2^i\alpha^{(0,i)}
      =\sum_{k=0}^{2n-2}2^k\sum_{\substack{0\le i,j<n\\i+j=k}}p_{ji}.
\end{align}
We group bits of weight $2^k$ into column $k$. Using $3$-to-$2$ adders, we repeatedly replace three equally weighted bits by a sum and a carry, with half adders used where only two bits are combined. We arrange these operations in parallel until at most two bits remain in each column. The remaining bits form two $2n$-bit numbers, $S$ and $C$, which we add using one final binary adder to obtain $S+C=xy$.

\subsubsection{Full and Half Adders}\label{subsec:modified_qcla}

A full adder takes three equally weighted bits $a,b,c$ and produces a sum bit $s$ and a carry bit $q$, where
\begin{align}
    s&=a\oplus b\oplus c,\\
    q&=ab\oplus ac\oplus bc.
\end{align}
The carry is one when at least two inputs are one. Thus
\begin{align}
    a+b+c=s+2q.
\end{align}
If the inputs belong to column $k$, the sum bit remains in column $k$ and the carry bit goes to column $k+1$. Only these two output bits participate in the remaining addition.

We use the \textit{full adder} of Kim et al.~\cite{Kim2025}, which uses
one zero ancilla to implement
\begin{align}
    \ket{a,b,c,0}\longmapsto
    \ket{q,a\oplus b,a\oplus c,s}.
    \label{eq:fulladder}
\end{align}
The sum $s$ and carry $q$ are used in the next addition stage. We keep the other two bits so that we can later reverse the circuit. As shown in Fig.~\ref{fig:full_adder}, the circuit uses five CNOTs and one Toffoli, with depth four and Toffoli depth one.

\begin{figure}[h]
    \centering
    \resizebox{0.6\columnwidth}{!}{
    \begin{quantikz}[column sep=0.4cm]
        \lstick{$\ket{a}$}
        & \qw & \ctrl{3} & \qw & \ctrl{1} & \qw & \targ{}
        & \rstick{$\ket{q}$} \qw \\
        \lstick{$\ket{b}$}
        & \ctrl{1} & \qw & \qw & \targ{} & \ctrl{1} & \ctrl{-1}
        & \rstick{$\ket{a\oplus b}$} \qw \\
        \lstick{$\ket{c}$}
        & \targ{} & \qw & \ctrl{1} & \qw & \targ{} & \ctrl{-2}
        & \rstick{$\ket{a\oplus c}$} \qw \\
        \lstick{$\ket{0}$}
        & \qw & \targ{} & \targ{} & \qw & \qw & \qw
        & \rstick{$\ket{s}$} \qw
    \end{quantikz}
    }
    \caption{Reversible $3$-to-$2$ adder using five CNOTs and one
    Toffoli. The first two CNOTs can run in parallel, as can the
    next two, giving depth four and Toffoli depth one.}
    \label{fig:full_adder}
\end{figure}

When we need to combine only two bits $a,b$, we use a \textit{half adder}.
Its sum is $a\oplus b$ and its carry is $ab$, giving
\begin{align}
    \ket{a,b,0}\longmapsto\ket{a,a\oplus b,ab}.
    \label{eq:halfadder}
\end{align}
This is simply implementable using a CNOT and a Toffoli, and has depth $2$ and Toffoli depth $1$.

\subsubsection{Arranging the Addition Tree}\label{subsec:tree_schedule}

To arrange the addition tree, we first count the number of bits
in each column. Column $k$ contains the partial-product bits
$p_{ji}$ with weight $2^i2^j=2^k$. We call the number of bits
in a column its \emph{height}.

We use superscripts for stages and subscripts for columns: $h_k^{(r)}$ denotes the height of column $k$ after stage $r$,
and the initial heights are
\begin{align}
    h_k^{(0)}=n-|k-(n-1)|,\qquad 0\le k\le2n-2.
    \label{eq:initial_heights}
\end{align}
More explicitly, these heights are
\begin{align}
    1,2,\ldots,n-1,n,n-1,\ldots,2,1.
\end{align}
We use full and half adders to reduce each column to at most
two bits.

Each stage $r$ has a target height $d^{(r)}$. We set $d^{(0)}=n$, the initial maximum height.
Stage $r$ takes the bits counted by $h_k^{(r-1)}$ as inputs
and leaves $h_k^{(r)}\le d^{(r)}$ bits in each column.

To choose the targets, we first determine how much one stage can reduce each column. A full adder takes three bits from a column and returns one sum bit, reducing its height by two.
A half adder reduces its height by one. In both cases, the carry goes to the next column.

Let $c_k^{(r)}$ be the number of carry bits produced by column
$k$ during stage $r$. These enter column $k+1$, so column $k$
receives $c_{k-1}^{(r)}$ carry bits during the same stage.
We set $c_{-1}^{(r)}=0$ for every stage.

When choosing the adders, we must account for these incoming
carries so that the column does not exceed the target height after the stage is done.
Without any reduction, column $k$ would finish stage $r$ with
$h_k^{(r-1)}+c_{k-1}^{(r)}$ bits. The required reduction is
therefore
\begin{align}
    e_k^{(r)}
    =\max\bigl(0,h_k^{(r-1)}+c_{k-1}^{(r)}-d^{(r)}\bigr).
\end{align}
We use a full adder for each reduction of two bits and a half
adder if one more bit needs to be removed. Writing $f_k^{(r)}$
and $u_k^{(r)}$ for the numbers of full and half adders in
column $k$ during stage $r$, we have
\begin{align}
    f_k^{(r)}&=\lfloor e_k^{(r)}/2\rfloor,\qquad
    u_k^{(r)}=e_k^{(r)}\bmod2.
    \label{eq:compressor_counts}
\end{align}
Each adder produces one carry bit. The outgoing carry count
and the height after the stage are consequently
\begin{align}
    c_k^{(r)}&=f_k^{(r)}+u_k^{(r)},\notag\\
    h_k^{(r)}
    &=h_k^{(r-1)}-2f_k^{(r)}-u_k^{(r)}+c_{k-1}^{(r)}.
    \label{eq:height_update}
\end{align}

These counts depend only on the arrangement of the bits, so
we choose the adders before running the circuit. 
Each adder uses distinct bits present at the start of stage $r$.
The newly produced carries contribute to $h_k^{(r)}$, but
are only used as adder inputs in stage $r+1$.

We choose each target height to be at least two-thirds of the previous
target. The following claim shows that this allows us to reach the new target in one stage of parallel adders.

\begin{claim}
Suppose each column entering stage $r$ has height at most
$d^{(r-1)}$, and the target $d^{(r)}$ satisfies
\begin{align}
    d^{(r)}\le d^{(r-1)}
    \le\lfloor3d^{(r)}/2\rfloor.
\end{align}
Then disjoint full and half adders can reduce every column
to height at most $d^{(r)}$, with depth at most four and
Toffoli depth at most one.
\end{claim}

\begin{proof}
We show that the adders have enough input bits without using
any carries produced during the stage.

The carry bound holds initially since $c_{-1}^{(r)}=0$.
Inductively, we would like to show that $c_{k-1}^{(r)}\le d^{(r-1)}-d^{(r)}$.
Using $h_k^{(r-1)}\le d^{(r-1)}$, we obtain
\begin{align}
    e_k^{(r)}&\le2\bigl(d^{(r-1)}-d^{(r)}\bigr).
\end{align}
And thus,
\begin{align}
    c_k^{(r)}&=\lceil e_k^{(r)}/2\rceil
               \le d^{(r-1)}-d^{(r)}.
\end{align}

We also need to check that each column has enough bits for
its adders. If $e_k^{(r)}=0$, no adders are needed.
Otherwise, the carry bound gives
\begin{align}
    e_k^{(r)}
    &\le h_k^{(r-1)}+d^{(r-1)}-2d^{(r)}\notag\\
    &\le h_k^{(r-1)}-\frac{d^{(r)}}{2},
\end{align}
since $d^{(r-1)}\le3d^{(r)}/2$ from our initial assumption.
The input height also satisfies 
\[
h_k^{(r-1)}\le \frac{3d^{(r)}}{2},
\]
so 
\[
\frac{d^{(r)}}{2}\ge \frac{h_k^{(r-1)}}{3}.
\]
This gives
\[
e_k^{(r)}\le\frac{2h_k^{(r-1)}}{3}.
\]

A full adder uses three input bits, while a half adder uses two.
The chosen adders therefore use
\begin{align}
    3f_k^{(r)}+2u_k^{(r)}
    =\left\lceil\frac{3e_k^{(r)}}{2}\right\rceil
    \le h_k^{(r-1)}
\end{align}
bits. Thus we have enough bits at the start of the stage to
give every adder its own inputs.

We can assign distinct input bits and one zero ancilla to
every adder. All adders run in parallel, giving
the claimed depths.
\end{proof}

We use this result to choose the targets by working backwards. The last stage must leave at most two bits per column, so it can start with height at most $\lfloor3\cdot2/2\rfloor=3$.
The stage before that can start with height at most
$\lfloor3\cdot3/2\rfloor=4$, and so on.

This gives Dadda's sequence for the target heights in each stage \cite{Dadda1965}:
\begin{align}
    2,3,4,6,9,13,\ldots.
\end{align}
We only need to consider the sequence up until the last value strictly less than $n$ in decreasing order as $d^{(1)},d^{(2)},\ldots$. 
The number of reduction stages is therefore the lowest value of $r$ such that $d^{(r)} \leq 2$, meaning the number of bits per column is at most $2$. For $n\le2$, no reduction stage is needed.

For example, when $n=4$, we have $d^{(0)}=4$, $d^{(1)}=3$, and $d^{(2)}=2$:
\begin{align}
    &(1,2,3,4,3,2,1)\notag\\
    &\qquad\xrightarrow{r=1}(1,2,3,3,3,3,1)\notag\\
    &\qquad\xrightarrow{r=2}(1,2,2,2,2,2,2).
\end{align}
In stage one, a half adder reduces column three (zero-indexed) from four bits to three and produces one carry bit, so $c_3^{(1)}=1$. Column four already contains three bits, so we also apply a half adder to two of its original bits to leave room for this carry. Both half adders run in parallel because neither uses the other's output.

Let $R(n)$ denote the number of reduction stages for $n$-bit multiplication. Starting from height two and working backwards, the target height after $t$ steps is at least $1+(3/2)^t$. This reaches $n$ once $t\ge\log_{3/2}(n-1)$, so
\begin{align}
    R(n)\le\left\lceil\log_{3/2}(n-1)\right\rceil
    \qquad(n\ge2).
    \label{eq:stage_bound}
\end{align}
Each stage has Toffoli depth at most one, so the forward reduction and its uncomputation each have Toffoli depth at most $R(n)$.

After the last stage, at most two bits remain in each column. We place one in $S$ and the other in $C$, using zero qubits for missing positions. 
Every full and half adder has preserved the weighted sum, so these two $2n$-bit numbers satisfy $S+C=xy$.

It remains to add these two $2n$-bit numbers together using the out-of-place Quantum Carry Lookahead Adder (QCLA) of Draper et al.~\cite{Draper2004} to implement
\begin{align}
    \ket{S}\ket{C}\ket{0}^{\otimes2n}
    \longmapsto\ket{S}\ket{C}\ket{S+C\bmod2^{2n}}.
    \label{eq:finalqcla}
\end{align}
Since $xy<2^{2n}$, the output is the exact product.

We then run the reduction stages of the \texttt{Parallel Adder Tree} in reverse to uncompute all intermediary registers in the same $T$-depth.

\subsubsection{Resource Analysis}\label{subsec:efficient_adder}

For the resource estimates below, we assume $n\ge4$ is a power of two. Each full and half adder uses one additional qubit. We first count these adders, then determine the space and depth required.

\begin{claim}
For $n\ge3$, the tree uses $(n-1)(n-3)$ full adders and $n-1$ half adders, requiring $(n-1)(n-2)$ additional qubits.
\end{claim}
\begin{proof}
The final stage leaves one bit in column zero and two bits in each column from one through $2n-2$, giving $4n-3$ bits. Each full adder reduces the total by one, while a half adder leaves it unchanged. Thus the number of full adders is
\begin{align}
    n^2-(4n-3)=(n-1)(n-3).
\end{align}
Counting the columns with odd $e_k^{(r)}$ using Eq.~\eqref{eq:height_update} gives $n-d^{(1)}+1$ half adders in the first stage and $d^{(r-1)}-d^{(r)}$ in each later stage. Since $d^{(R(n))}=2$, their total is
\begin{align}
    n-d^{(1)}+1+
    \sum_{r=2}^{R(n)}\bigl(d^{(r-1)}-d^{(r)}\bigr)=n-1.
\end{align}
Each adder requires one additional qubit, giving
\begin{align}
    (n-1)(n-3)+(n-1)=(n-1)(n-2).
    \label{eq:compression_space}
\end{align}
\end{proof}

The final two rows require three zero padding qubits, and the QCLA requires fewer than $2n$ prefix qubits \cite{Draper2004}. Including the inputs, product register, partial products, and all retained adder outputs, the addition stage fits within
\begin{align}
    4n+n^2+(n-1)(n-2)+3+2n=2n^2+3n+5
\end{align}
qubits.

The forward and reverse trees together have $2R(n)$ stages. Each has depth at most four and Toffoli depth at most one. Including the final QCLA \cite{Draper2004}, for power-of-two $n\ge4$, we obtain
\begin{align}
    D_{\mathrm{tree}}&\le8R(n)+2\log_2 n+6,\\
    D_{\mathrm{Tof,tree}}&\le2R(n)+2\log_2 n+3.
\end{align}

\section{Fast Quantum Multiplication Algorithm}

In this section, we combine the submodules introduced in Section~\ref{sec:prelim} to construct the Fast Quantum Multiplication Algorithm. We then give its Clifford+$T$ implementation and resource costs.

\begin{algorithm}[t]
\caption{fast\_quantum\_multiplication\_algorithm}\label{algo:multiplication}
\KwIn{$n\ge3$, $\ket{x}\ket{y}\ket{0}^{\otimes2n}\ket{0}^{\otimes A}$, where $A=4n^2-2n$}
\texttt{fast\_copy}$(\ket{x})$ and \texttt{fast\_copy}$(\ket{y})$ in parallel\;
\texttt{partial\_products} on the copied inputs\;
Uncompute \texttt{fast\_copy}; reuse the cleared copy wires during addition\;
\texttt{parallel\_adder\_tree}, including the final QCLA and reverse compression\;
Redo \texttt{fast\_copy} on both inputs\;
Uncompute \texttt{partial\_products}\;
Uncompute \texttt{fast\_copy} on both inputs\;
\KwOut{$\ket{x}\ket{y}\ket{xy}\ket{0}^{\otimes A}$}
\end{algorithm}

\begin{claim}
For $n\ge3$, Algorithm~\ref{algo:multiplication} implements Eq.~\eqref{eq:Desired} in $O(\log n)$ circuit depth, using $O(n^2)$ gates and ancilla qubits.
\end{claim}
\begin{proof}
We start with $\ket{x}\ket{y}\ket{0}^{\otimes2n}\ket{0}^{\otimes A}$ and follow the seven steps of the algorithm.

\paragraph{Step 1: Fast Copy of $\ket{x}$ and $\ket{y}$.}
We apply Algorithm~\ref{algo:fastcopy} to both inputs in parallel. This produces $n$ computational-basis copies of each input in $\lceil\log n\rceil$ depth. Including the originals, the copies occupy $2n^2$ qubits; the additional space is $2n(n-1)$.

\paragraph{Step 2: Compute All Partial Products.}
We apply \texttt{conditional\_copy} for every $y_i$ in parallel using Algorithm~\ref{algo:itwopowermultiples}. All $n^2$ Toffolis act on disjoint wires and form one layer. Their outputs are $\alpha^{(0,i)}=y_ix$, occupying $n^2$ qubits.

\paragraph{Step 3: Uncompute Step 1.}
The Toffolis have not changed their control bits. We therefore reverse the copying networks and return all additional copies to zero. 
The partial products remain available, and the cleared wires are reused in the next step.

\paragraph{Step 4: Sum of Partial Products.}
We apply \texttt{parallel\_adder\_tree}. The compression stages preserve the weighted sum in Eq.~\eqref{eq:DefnAlpha}, and the final QCLA writes that sum, $xy$, into the product register. 
Uncomputing the reductions restores every partial product and clears the additional working qubits needed for the adders. Grouping registers for notation, the state is now
\begin{align}
    \ket{x}\ket{y}\ket{xy}
    \bigotimes_{i=0}^{n-1}\ket{\alpha^{(0,i)}}
    \ket{0}^{\otimes(A-n^2)}.
\end{align}

\paragraph{Step 5: Restore Copies of $\ket{x}$ and $\ket{y}$.}
We recreate the $n$ copies of each input on the same physical wires used in Step 1. These wires are zero following Step 4. This again has depth $\lceil\log n\rceil$.

\paragraph{Step 6: Uncompute Partial Products.}
We reverse Algorithm~\ref{algo:itwopowermultiples}. The input copies have been restored, so the inverse clears every partial-product bit in one Toffoli layer.

\paragraph{Step 7: Uncompute the Copies.}
Finally, we reverse the copying networks and return all remaining helper qubits to zero. The final state is
\begin{align}
    \ket{x}\ket{y}\ket{xy}\ket{0}^{\otimes A}.
\end{align}
The argument holds on each computational-basis branch. By linearity, the same circuit works for arbitrary superpositions and inputs entangled with an external register.

There are four logarithmic-depth copying operations, two partial-product layers, $2R(n)$ constant-depth compression stages, and one logarithmic-depth QCLA. Their total depth is $O(\log n)$. The counts in the preceding subsections give $O(n^2)$ gates and qubits, including uncomputation.
\end{proof}

\subsection{A Low-$T$-Depth Implementation}\label{subsec:phase_implementation}

We now provide a $T$-depth optimized implementation of Algorithm~\ref{algo:multiplication}. Each Toffoli operation is replaced by the relative-phase gate \cite{Selinger2013}
\begin{align}
    \mathcal R\ket{a,b,c}
    =(-i)^{ab}\ket{a,b,c\oplus ab}.
    \label{eq:relative_toffoli}
\end{align}
This gate uses four $T$ or $T^\dagger$ gates in one layer and one clean ancilla, which is returned to zero. It performs the same bit operation as a Toffoli, but adds a phase of $-i$ when both controls are one. We must either cancel or correct this phase.

\paragraph{Partial products and the addition tree.}
We replace each forward Toffoli in these stages by $\mathcal R$ and the inverse by $\mathcal R^\dagger$. The QCLA leaves the intermediate registers unchanged while writing the product into a separate register. Running the previous stages in reverse therefore cancels their phases and clears the workspace.

\paragraph{The final QCLA.}
Since the final adder is not reversed, we correct its phases using the Clifford gate $\mathsf S=\operatorname{diag}(1,i)$. Its Toffoli operations have two forms.

For an AND computed into zero, apply $\mathcal R$ followed by $\mathsf S$ on the target. The phases cancel:
\begin{align}
    \ket{a,b,0}
    \longmapsto(-i)^{ab}i^{ab}\ket{a,b,ab}
    =\ket{a,b,ab}.
\end{align}
For uncomputation, apply $\mathsf S^\dagger$ followed by $\mathcal R^\dagger$.

For a carry update $g\gets g\oplus uv$, the QCLA ensures that $g$ and $uv$ cannot both be one \cite{Draper2004}. This is because a block cannot both generate and propagate a carry. We therefore have $g\oplus uv=g+uv$.

We apply $\mathsf S^\dagger$ to the target before $\mathcal R$ and $\mathsf S$ afterwards. The phases cancel:
\begin{align}
    i^{-g}(-i)^{uv}i^{g\oplus uv}
    =i^{-g}i^{-uv}i^{g+uv}=1.
\end{align}

These corrections use only Clifford gates, so they add no $T$-depth. With one additional ancilla for each parallel operation, each Toffoli layer takes one $T$ layer. All ancillas are returned to zero, and the complete circuit implements Eq.~\eqref{eq:Desired} exactly.

\subsection{Resource Costs}

We now combine the costs of the subroutines, including uncomputation. We assume $n\ge4$ is a power of two.

Computing the partial products uses $n^2$ Toffolis and clearing them uses another $n^2$. The reduction portion of the adder tree uses $(n-1)(n-2)$ adders in total, as discussed in Subsec.~\ref{subsec:efficient_adder}. Each adder uses one Toffoli and so the forward and reverse reduction stages use $2(n-1)(n-2)$ Toffolis. Adding the final QCLA \cite{Draper2004} gives
\begin{align}
    N_{\mathrm{Tof}}
    &=2n^2+2(n-1)(n-2)\notag\\
    &\quad+10n-9-6\log_2 n\notag\\
    &=4n^2+4n-6\log_2 n-5.
    \label{eq:Tofcount}
\end{align}

The four copying operations on both inputs use $8n(n-1)$ CNOTs. Each full adder uses five CNOTs and each half adder uses one. Including their inverses and the final QCLA,
\begin{align}
    N_{\mathrm{CNOT}}
    &=8n(n-1)\notag\\
    &\quad+2\bigl[5(n-1)(n-3)+(n-1)\bigr]\notag\\
    &\quad+6n-2\notag\\
    &=18n^2-40n+26.
\end{align}
The total number of gates in our construction before decomposition is therefore
\begin{align}
    N_{\mathrm{Tof}}+N_{\mathrm{CNOT}}
    =22n^2-36n-6\log_2 n+21.
    \label{eq:nativecount}
\end{align}

The copying workspace is reset before addition and reused for the adder ancillas. As discussed in Subsec.~\ref{subsec:efficient_adder}, the adder tree (forward and reverse) needs $(n-1)(n-2)$ ancillas, and so the $2n^2-2n$ ancillas freed by uncomputing the copying stage is sufficient. Additionally, we require $n^2$ ancillas for storing the partial products and $n^2$ ancillas to run at most $n^2$ relative-phase Toffoli gates in parallel. In total, our construction requires
\begin{align}
    A=4n^2 - 2n
\end{align}
ancilla qubits.

The four copying operations contribute $4\log_2 n$ depth, and computing and clearing the partial products takes two Toffoli layers. Including the addition tree and final QCLA,
\begin{align}
    D&\le8R(n)+6\log_2 n+8,\\
    D_{\mathrm{Tof}}&\le2R(n)+2\log_2 n+5.
    \label{eq:fulldepth}
\end{align}

Our Clifford+$T$ implementation uses four $T$ gates per Toffoli operation and one $T$ layer per Toffoli layer. Thus, including uncomputation, the $T$-count and $T$-depth are
\begin{align}
    N_T&=4N_{\mathrm{Tof}}
        =16n^2+16n-24\log_2 n-20,\label{eq:Tcount}\\
    D_T&\le2R(n)+2\log_2 n+5\notag\\
       &\le2\left\lceil\log_{3/2}(n-1)\right\rceil
          +2\log_2 n+5,
    \label{eq:Tdepth}
\end{align}
respectively. The gate counts are exact for this implementation, while the depths are upper bounds. It follows that our construction has $O(\log n )$ depth and $T$-depth, and uses $O(n^2)$ Clifford + $T$ gates and ancilla qubits.

\section{Conclusion}

Optimizing quantum algorithms for multiplication is important for their role in applications such as modular exponentiation in Shor's algorithm, Hamiltonian simulation using truncated Taylor series, and algorithms for linear systems \cite{Shor1997,pavlidis2013fastquantummodularexponentiation, Harrow_2009,Childs_2017,Berry_2015}. Improvements in depth, gate count, and ancillary qubit usage can therefore reduce the resources required by higher-level quantum algorithms.

In this work, we presented a quantum multiplication algorithm with $O(\log n)$ circuit depth and $T$-depth, using $O(n^2)$ gates and qubits. We compute all partial products in parallel and add them using a tree of reversible $3$-to-$2$ adders followed by a QCLA. The copying workspace is reused during addition, and uncomputation returns all helper ancillas to zero while preserving both inputs.

Using relative-phase gates and Clifford phase corrections, we implement each Toffoli layer in one $T$ layer. Compared to the reported construction of Kim et al.~\cite{Kim2025}, we achieve lower $T$-depths at the cost of additional ancillas and $T$ gates. Our construction also establishes a lower $T$-depth--space volume for $n=32,64$, using a fully unitary circuit.

Our construction assumes all-to-all connectivity. Future work may explore the overhead of restricted connectivity, further reductions in the final-adder depth, alternative workspace schedules, and different fault-tolerant implementations.

\begin{acknowledgments}
Fred Sun thanks Richard Cleve for pointing him to \cite{Gossett1998}. The authors thank Matthew Amy, Vlad Gheorghiu, and Andrew Jena for their feedback on this manuscript.
\end{acknowledgments}

\section*{Use of Generative AI}

The authors used OpenAI's ChatGPT to assist with editing the manuscript and resource analysis. The authors are responsible for the final results and conclusions.

\bibliography{bib}

\end{document}